\documentclass{article}

\usepackage{amsmath}

\usepackage{amsthm}
\usepackage{graphicx}
\usepackage[colorlinks=true, allcolors=blue]{hyperref}
\usepackage{subcaption}
\usepackage{graphicx}
\usepackage{caption}
\usepackage{comment}

\newtheorem{claim}{Claim}

\newtheorem{definition}{Definition}
\newtheorem{lemma}{Lemma}
\newtheorem{theorem}{Theorem}
\newtheorem{corollary}{Corollary}

\title{Effects of graph operations on star pairwise compatibility graphs}
\author{Angelo Monti\inst{1}\orcidID{0000-0002-3309-8249} \and
Blerina Sinaimeri\inst{2}\orcidID{0000-0002-9797-7592} }

\author{A. Monti\footnote{Computer Science Department, Sapienza University of Rome, Italy,
\href{mailto:monti@di.uniroma1.it}{monti@di.uniroma1.it}}
 \and B. Sinaimeri\footnote{Luiss University, Rome, Italy. \href{mailto:bsinaimeri@luiss.it}{bsinaimeri@luiss.it}}}

\begin{document}
\maketitle

\begin{abstract}
A graph $G=(V,E)$ is defined as a star-$k$-PCG when it is possible to assign a positive real number weight $w$ to each vertex $V$, and define $k$ distinct intervals $I_1, I_2, \ldots I_k$, in such a way that there is an edge $uv$ in $E$ if and only if the sum of the weights of vertices $u$ and $v$ falls within the union of these intervals.  The star-$k$-PCG class is connected to two significant categories of graphs, namely PCGs and multithreshold graphs.  The star number of a graph $G$, is the smallest $k$ for which $G$ is a star-$k$-PCG. In this paper, we study the effects of various graph operations, such as the addition of twins, pendant vertices, universal vertices, or isolated vertices, on the star number of the graph resulting from these operations. As a direct application of our results,  we determine the star number of lobster graphs and provide an upper bound for the star number of acyclic graphs.
\end{abstract}

\maketitle

A $k$-pairwise compatibility graph $G$ (shortly $k$-PCG), also referred to as a multi-interval PCG \cite{ahmed17,calamoneri2022}, is a type of graph characterized by the existence of a non-negative edge-weighted tree $T$ and $k$ distinct intervals $I_1, I_2, \ldots, I_k$ of non-negative real numbers. In such a graph, each vertex of $G$ corresponds to a leaf of $T$, and an edge between two vertices in $G$ is present if the distance between their corresponding leaves in $T$ falls within $I_1 \cup I_2 \cup \ldots \cup I_k$. The tree $T$ serving this function is known as the $k$-witness tree. The concept of $1$-PCGs, also known as PCGs, originated from the problem of reconstructing phylogenetic trees \cite{KPM03} and it has proven valuable in analyzing rare evolutionary events, including horizontal gene transfer (see \textit{e.g.} \cite{Long2020}). 

Notice that for a $k$-PCG the witness tree is not unique. An example of a graph $G$ together with two different witness trees is provided in Figure~\ref{fig:example-pcg}. 
In this study, we focus on $k$-PCGs for which there exists a witness tree that is a star, and refer to these graphs as star-$k$-PCGs \cite{monti2024starkpcgs}. 

One of the main reasons to study star-$k$-PCGs is that, while determining in polynomial time whether a given graph is a $k$-PCG remains an open problem for any constant $k$ (including $k=1$), the problem appears to be more tractable for star-$k$-PCGs. In fact, polynomial-time algorithms for identifying graphs that are star-$1$-PCGs have already been proposed in the literature \cite{Xiao2020,Kobayashi22}. Nevertheless, for $k>1$, the problem remains open even for star-$k$-PCGs. 

Another reason for studying this class comes from the connection it has with multithreshold graphs \cite{Jamison20}. In \cite{Kobayashi22,monti2024starkpcgs} it is shown that the class of star-$k$-PCGs is equivalent to the class of $2k$-threshold graphs, which has gained considerable interest within the research community since its introduction in \cite{Jamison20}, as evidenced by the following studies \cite{Jamison20,Puleo2020,Jamison2021,chen2022,Kittipassorn2024}. These connections and subsequent studies highlight the significance the significance and  applicability of these graph classes in understanding complex evolutionary processes. 

To study star-$k$-PCGs, it is natural to define the \emph{star number} of a graph $G$, denoted as $\gamma(G)$, which is the smallest positive integer $k$ such that $G$ is a star-$k$-PCG \cite{monti_ictcs23,monti2024starkpcgs}. 

\begin{figure*}
\centering
\includegraphics[scale =0.37]{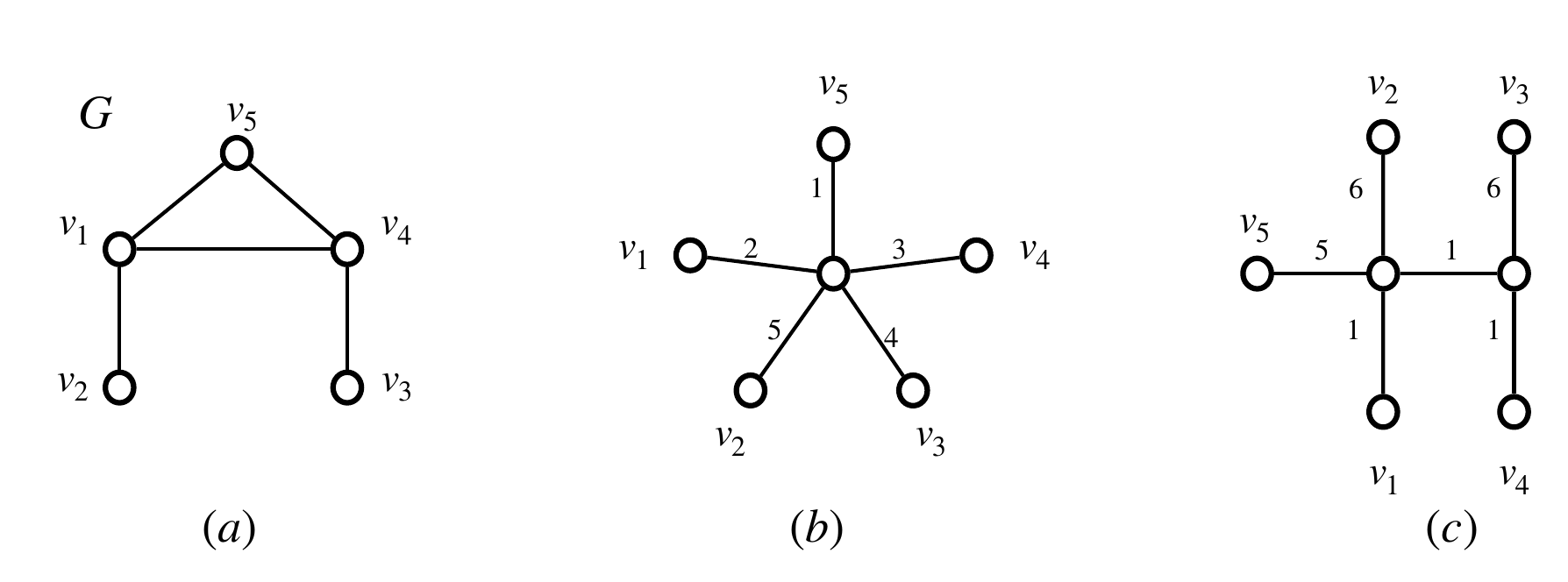} 
\caption{$(a)$ a graph $G$;  $(b)$ a star $2$-witness tree for $G$ with the corresponding intervals $I_1=[3,5]$ and $I_2=[7,7]$; $(c)$ a $1$-witness tree for $G$ with the corresponding interval $I_1=[3,7]$.}\label{fig:example-pcg}
\end{figure*}



In this study, we determine the star number of various simple graph classes. Our approach involves analyzing how specific graph operations impact the star number. This method is chosen because many graph classes arise from simpler ones through these operations. For instance, Xiao et al. \cite{Xiao2020op} have explored the impact of different graph operations on PCGs. Building on this research, we examine the effects of certain graph operations on the star number in star-PCGs. Specifically, we consider the following operations: adding universal and isolated vertices, incorporating true/false twins, appending a pendant vertex, and complementing a graph.

As an application of our results we determine the star number of lobsters and provide an upper bound for the star number of acyclic graphs. More specifically, we show that the star number of any lobster of radius at least  three, is 2, while the star number of any acyclic graph is at most its radius number. This represents an advancement beyond the well-known  general result stating that for any graph $G=(V,E)$,  $\gamma(G) \leq |E|$ \cite{ahmed17}. A direct corollary of our findings is that the multithreshold number \cite{Jamison20} of lobsters ranges between $3$ and $4$, and for acyclic graphs, it is at most twice their radius number. This is particularly interesting because obtaining good bounds on the multithreshold number is notoriously difficult. Indeed, for acyclic graphs, the exact answer is known only for caterpillars \cite{Jamison20}.

\begin{figure*}
\centering
\includegraphics[width=\linewidth]{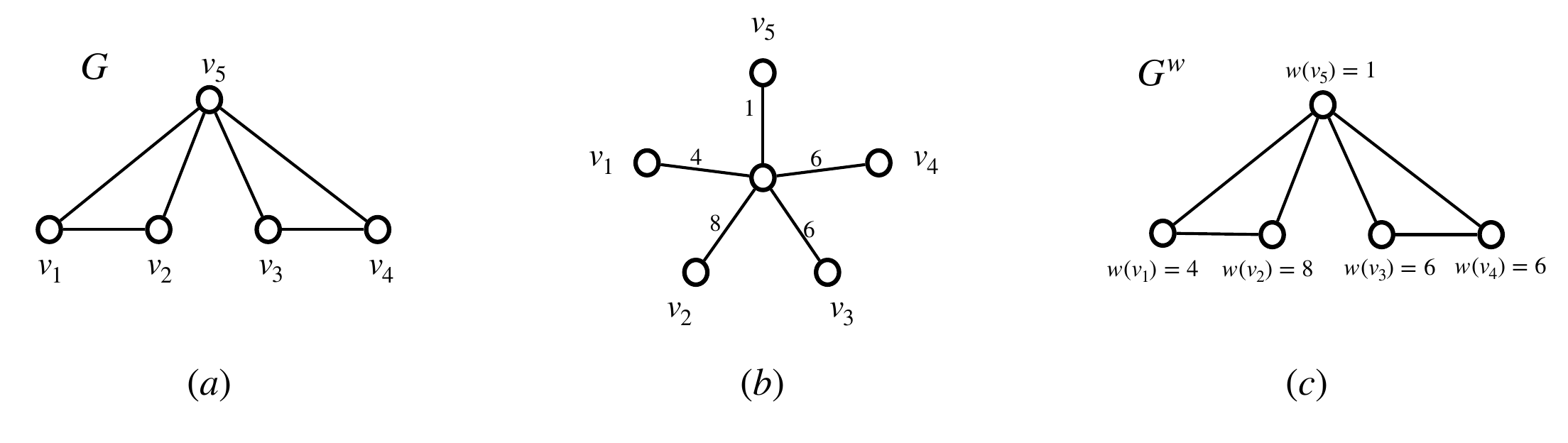} 
\caption{$(a)$ an example of a star-$2$-PCG $G$;  $(b)$ a star witness tree with the corresponding intervals $I_1=[5,9]$ and $I_2=[12,12]$; $(c)$ the witness graph $G^w$ and the corresponding intervals $I_1=[5,9]$ and $I_2=[12,12]$.}\label{fig:example-star}
\end{figure*}

\section{Preliminaries}\label{sec:Preliminaries}
In this paper we  only consider simple graphs, that is graphs that contain no loops or multiple edges. We focus  on undirected graphs and  for simplicity, we use a notational shorthand, writing $uv$ to represent the unordered pair  $\{u,v\}$.  For a graph $G=(V,E)$ and a vertex $u\in V$, the set $N(u)=\{v: uv\in E\}$ is called the \emph{neighborhood} of $u$.  

Two distinct vertices $u$ and $v$ in $G$ are called \emph{true twins} if $N(u)\cup\{u\} = N(v)\cup \{v\}$, and \emph{false twins} if $N(u)=N(v)$. 

A vertex $u$ of a graph $G=(V,E)$ is said \emph{isolated} if it has degree 0, \textit{i.e.} $|N(u)|=0$; is said \emph{universal} if it has degree $|V|-1$, \textit{i.e.} $|N(u)|=|V|-1$; and is said \emph{pendant} if it has degree 1, \textit{i.e.} $|N(u)|=1$. For a pendant vertex $v$ we will denote by $p(v)$ the only vertex adjacent to $v$. A pendant edge is an edge that is incident to a pendant vertex.

For any integer $n \geq 1$ we denote by $P_n$ the path on $n$ vertices.

 A \emph{caterpillar} is a tree in which all the vertices are within distance 1 of a central path. The central path contains only vertices of degree at least 2 (\textit{i.e.} vertices that are not leaves).

A \emph{lobster} is a graph such that when we delete its leaves, we obtain a caterpillar .

The \emph{distance} between two vertices $u$ and $v$ is the length (number of edges) of a shortest path between $u$ and $v$. 
For a connected graph the \emph{eccentricity} of a vertex $v$, denoted as $ecc(v)$, is defined as the maximum distance from $v$ to any other vertex in the graph. The  \emph{radius} of a graph $G$, denoted as $rad(G)$, is the minimum eccentricity among all the vertices in $G$. A vertex of minimum eccentricity is called \emph{center} of the graph.  For a disconnected graph $G$, we use the convention that the eccentricity (radius) of $G$ is the maximum eccentricity (radius) of its connected components.

For star-$k$-PCGs, the distance between any two vertices in the witness star tree is the sum of the weights of the two edges incident to these vertices. Therefore, it is possible to eliminate the star tree and use a definition that employs weights on the vertices of the graph (see Figure~\ref{fig:example-star}). Thus, throughout the paper, we will use the following equivalent definition of star-$k$-PCGs.

\begin{definition}
A graph $G=(V,E)$ is a star-$k$-PCG if there exists a weight function $w: V \rightarrow \mathcal{R}^+$ and $k$ mutually exclusive intervals $I_1, I_2, \ldots I_k$, such that there is an edge $uv \in E$ if and only if $w(u)+w(v) \in \bigcup_i I_i$ \cite{monti2024starkpcgs}. Such a vertex weighted graph $G^w$ is called the \emph{$k$-witness} of $G$.
\end{definition}

For a $k$-witness graph $G^w$ and for each vertex $v$,  $w(v)$ denotes the weight of the vertex $v$ and for each pair of vertices $u$ and $v$  of $G$ we denote by $w(uv)$ the sum $w(u)+w(v)$. 





Given a $k$-witness graph $G^w$, for the sake of simplicity we will always assume that the  $k$ disjoint intervals $I_1=[a_1,b_1], \ldots I_k=[a_k, b_k]$ are ordered in increasing order, that is for all $1\leq i \leq k-1$, $b_i < a_{i+1}$.

\begin{claim}\label{claim:non_edge_between_intervals}
 Let $G$ be a graph with star number $k>1$ and let $G^w$ be a $k$-witness graph and $I_i=[a_i, b_i]$ for $1\leq i \leq k$ the corresponding intervals.  Then for any $1\leq i \leq k-1$ there exist a non-edge $e \not \in E$ such that $b_i < w(e) < a_{i+1}$. 
\end{claim}
\begin{proof}
If there exists an $1\leq i \leq k-1$ for which the condition of the claim does not hold, then the two intervals $I_i, I_{i+1}$ could be merged in a single interval $[a_i, b_{i+1}]$ and $G^w$ would be a $(k-1)$-witness, contradicting the hypothesis that the star number of $G$ is $k$.  
\end{proof}

Consider a graph $G$ with star number $k$ and let $G^w$ be a $k$-witness  of $G$, with intervals $I_i=[a_i, b_i]$ for $1\leq i \leq k$. 
We say that $G^w$ is \emph{left-free} (similarly, \emph{right-free}) if for every non-edge $e$ we have $w(e)> b_1$ (similarly, $w(e) < a_k$)).



Notice that the $2$-witness graph depicted in Figure~\ref{fig:example-star}b is left-free, while the one in Figure~\ref{fig:example-star}c is both left and right free. Finally, in Figure~\ref{fig:universal_a} we provide an example of a witness graph that is not free.  

\begin{lemma}\label{lem:left-right-unbalanced}
Let $G$ be a graph with $\gamma(G)=k$, then there exists a left-free $k$-witness $G^w$ if and only if there exists a right-free $k$-witness $G^{w'}$.
\end{lemma}
\begin{proof}
Let $G=(V,E)$ be a graph  with star number $k$  and  let $G^w$ be a $k$-witness that is right-free with  $I_i=[a_i, b_i]$ for $1\leq i \leq k$ the corresponding intervals. If $G^w$ is also right-free, the proof follows. Otherwise notice that for all pairs of vertices $u,v$ (that not necessarily correspond to edges) we have $w(uv)\leq b_k$ and thus for any $u\in V$, $w(u) \leq b_k$. We define now a left-free $G^{w'}$ as follows. For all $u \in V$ we set $w'(u)=b_k-w(u)$ and for any $1\leq i \leq k$ we define  $I'_{k+1-i}=[2b_k-b_i,2b_k-a_i]$. Notice that all the new weights and intervals are positive. Moreover, the new intervals $I'_i$ are ordered in increasing order. Now observe that for any edge $uv \in E$ there exists an $i$ such that $a_i\leq w(uv) \leq b_i$. Then $2b_k-b_i\leq w'(uv)=2b_k-w(uv)\leq 2b_k-a_i$ and thus $w'(uv)\in I'_i$. Next, for any $uv \not \in E$, as $G^w$ is right-free we have that either $w(uv)< a_1$ or there exists and $1\leq i \leq k-1$ for which $b_i< w(uv)< a_{i+1}$. Then we have $w'(uv)=2b_k-w(uv)$ and thus either $w'(uv)> 2b_k-a_1$ or $2b_k-a_i< w'(uv)< 2b_k-b_{i+1}$. Thus,  $w'(uv) \not \in \cup_i I'_i$. Clearly, this proves that $G^{w'}$ is a left-free $k$-witness. 

Similarly, we can show that if $G^w$ is left-free, then it is possible to construct a $G^{w'}$ that is right-free.
\end{proof}

Based on Lemma~\ref{lem:left-right-unbalanced} from now on we will simply use \emph{free witness} without specifically mentioning 'left' or 'right'. The next result shows that there exist graphs $G$ that do not have a free $\gamma(G)$-witness.

\begin{theorem}\label{theo_path_balanced}
Let $G$ be a graph with $\gamma(G)=1$ that contains $P_4$ as an induced subgraph. Then every $1$-witness of $G$ is not free.
\end{theorem}
\begin{proof}
 To prove the claim, it is enough to prove that $P_4$ has no free $1$-witness. Indeed, if we assume on the contrary that there exists a free $1$-witness ${P_n}^w$, then as $\gamma(G)=\gamma(P_4)=1$ it is not difficult to see that a subgraph of ${G}^w$ would serve as a free $1$-witness for $P_4$.  Therefore, we proceed to prove that $P_4$ has no free $1$-witness. Let $P_4=v_1, v_2, v_3, v_4$ where $v_i v_{i+1} \in E(P_4)$  with $1\leq i \leq 3$.  Consider any $1$-witness $P_4^w$ with corresponding interval $I_1=[a_1,b_1]$. \emph{W.l.o.g.} assume $w(v_1)< w(v_4)$. Then we have $a_1 \leq w(v_1v_2) < w(v_1v4)$ and thus necessarily it must hold $w(v_1v_4) > b_1$. Moreover $b_1 \geq w(v_3v_4) > w(v_3v1)$ and we must have $w(v_3v1)< a_1$. Thus,  $P_4^w$ cannot be free. 
\end{proof}

In \cite{Xiao2020,monti2024starkpcgs} it is shown that $\gamma(P_n)=1$  and thus we have the following
\begin{corollary}
Every $1$-witness of $P_n$, with $n\geq 4$, is not free. 
\end{corollary}\label{corollary_path_balanced}

We  define a normal form for the $k$-witness graphs. 

\begin{definition}\label{def:normal_form}
Let $G=(V,E)$ be a graph with $\gamma(G)=k$. A $k$-witness $G^w$ is said to be in \emph{normal form} if for any vertex $x\in V$, both of the followings are true
 \begin{itemize}
     \item[(i)]  for any vertex $u\neq x$ in $V$, it holds $w(x)\neq w(u)$. 
\item[(ii)] for any two distinct vertices $u_1$ and $u_2$ in $V$ it holds $2 w(x)\neq w(u_1) + w(u_2)$
 \end{itemize}
\end{definition}

The next lemma shows that every graph $G$ has a $\gamma(G)$-witness in normal form.

\begin{lemma}\label{lem:distinct_weight}
 Let $G=(V,E)$ be a graph with $\gamma(G)=k$, then there exists a $k$-witness $G^w$ in normal form.
 
\end{lemma}
\begin{proof} 
Notice that in the extreme case where the graph $G$ is complete (or empty), the claim trivially holds as it is always possible to  assign different weights to the vertices and set the single interval $I$ appropriately in such a way both (i) and (ii) are satisfied. 
So we consider now graphs that have at least one edge and one non-edge. Let $G=(V,E)$ be a graph with star number $k$ and let $G^w$ a $k$-witness graph and $I_i=[a_i, b_i]$ for $1\leq i \leq k$ the corresponding intervals. Let $x$ be a vertex in $V$ for which at least one among (i) and (ii) does not hold. If no such vertex $x$ exists, we are done. However, if $x$ does exist, we adjust the weight of $x$ and the intervals $I_i$ to ensure that, within this updated $k$-witness configuration, $x$ now meets both conditions (i) and (ii). Furthermore,  vertices which already met these conditions continue to do so. By repeating this process, we consistently reduce the number of vertices that violate the conditions in (i) and (ii), and thus the proof will follow. We define the following three quantities:
$$
 \delta_1= \min_{\substack{xy \notin E \\ 1 \leq i \leq k}} \left\{ \left| w(xy) - b_i \right|, \left| w(xy) - a_i \right| \right\} >0
$$
Notice that for every non-edge $xy$, its weight is at least $\delta_1$ apart from the endpoints of each of the intervals $I_i$. Moreover, $\delta_1 >0$ as there is at least one non edge in $G$. 
$$
 \delta_2= \min_{\substack{y \in V \\ w(y) \neq w(x)}} \left\{  \left| w(x) - w(y)\right| \right\} >0
$$
Notice that all weights different from $w(x)$ differ from it  by at least $\delta_2$. This will be useful in proving that condition (i) is met for $x$. Moreover,  $\delta_2 >0$ as the only case where all the vertices have equal weight is where the graph is complete or empty. 
$$
 \delta_3= \min_{\substack{u_1, u_2, u_3 \in V\\ x\in \{u_1,u_2,u_3\}; 2w(u_1) \neq w(u_2u_3)}} \left\{  \left| 2w(u_1) - w(u_2u_3))\right| \right\}
$$
The way we defined $\delta_3$ will help proving that condition (ii) is met for $x$.
Finally, we set $\epsilon=\frac{1}{4} \min \{\delta_1,\delta_2, \delta_3\}$ if $\delta_3 > 0$, otherwise $\epsilon=\frac{1}{4} \min \{\delta_1,\delta_2\}$. 

Then define $w': V \rightarrow \mathcal{R}^+$ as follows
$$
   w'(v) = 
    \begin{cases} 
     w(v)+\epsilon  & \text{if } v =x, \\
      w(v) & \text{otherwise } 
    \end{cases}
$$
For $1 \leq i \leq k$ we set $I_i'=[a_i, b_i+\epsilon]$. Notice that as $\epsilon< \delta_1$, from Claim~\ref{claim:non_edge_between_intervals}, we have that the intervals $I'_i$ do not overlap. Moreover, by construction it is not difficult to see that  $x$ satisfies both condition (i) and (ii).  In addition, as we changed only the weight of $x$ then all the vertices $y$ that satisfied both conditions (1) and (2) still continue to satisfy them.

It remains to prove that $G^{w'}$ is a $k$-witness with respect to the intervals $I_1', \ldots, I_k'$. Given two vertices $u$ and $v$ there are two cases to consider: 

\begin{itemize}
    \item $u\neq x$ and $v\neq x$. Then as $\epsilon< \delta_1$ it is easy to see that $ w(uv) \in \cup_i I_i$ if and only if $ w'(uv) \in \cup_i I_i'$.
\item $u=x$. Then we have $w'(xv)=w(xv)+\epsilon$. Clearly, if $ w(xv) \in I_i$ for some $1\leq i \leq k$, then $a_i \leq w(xv) \leq b_i$. Thus we have $a_i \leq w'(xv)=w(xv) +\epsilon \leq b_i+\epsilon$. Thus $w'(xv)\in I'_i$. 
Otherwise if $ w(xv) \not\in \cup_i I_i$ since $\epsilon < \delta_1$ we have that $ w'(xv)=w(xv)+\epsilon \not\in \cup_i I_i$. 
\end{itemize} 
This concludes the proof.
\end{proof}

It is not difficult to see that the proof of Lemma~\ref{lem:distinct_weight} preserves the free property of a witness, that is, if the graph $G^w$ is a left-free (or right-free) $k$-witness then $G^{w'}$ is a left-free (or right-free) $k$-witness in normal form. Thus we have the following.

\begin{corollary}\label{cor:normal_form_free}
For any graph $G$ with $\gamma(G)=k$, if there exists a left-free (or right-free) $k$-witness, then there also exists a left-free (or right-free) $k$-witness in normal form.
\end{corollary}

\section{Graph operations on star-$k$-PCGs}\label{sec:operations}
In this section we consider how the star number of a graph changes after we perform one of the following operations:
\begin{enumerate}
    \item[(a)] Adding an isolated vertex;
    \item[(b)] Adding a universal vertex;
    \item[(c)] Adding a pendant vertex;
    \item[(d)] Adding a vertex that is a false twin for the old vertex \(v\);
    \item[(e)] Adding a vertex that is a true twin for the old vertex \(v\);
    \item[(f)] Graph complement;
\end{enumerate}

\begin{theorem}
Given a graph $G$ with $\gamma(G) = k$, let $G'$ be the graph obtained from $G$ by adding an isolated vertex. It holds that $\gamma(G') = k$.
\end{theorem}
\begin{proof}
Let $G=(V,E)$ be a graph with star number $k$ and let $G^w$ be a $k$-witness graph with $I_i=[a_i, b_i]$ for $1\leq i \leq k$. Now consider the graph $G'$ obtained from $G$ by adding an isolated vertex $v$. We construct the $k$-witness $G'^{w'}$ from $G^w$ by assigning the weight $b_k+1$ to the new vertex $v$. It is straightforward to see that this is a $k$-witness for $G'$.
\end{proof}

\begin{figure*}
    \begin{subfigure}[b]{0.55\textwidth}
    \centering
    \includegraphics[width=\linewidth]{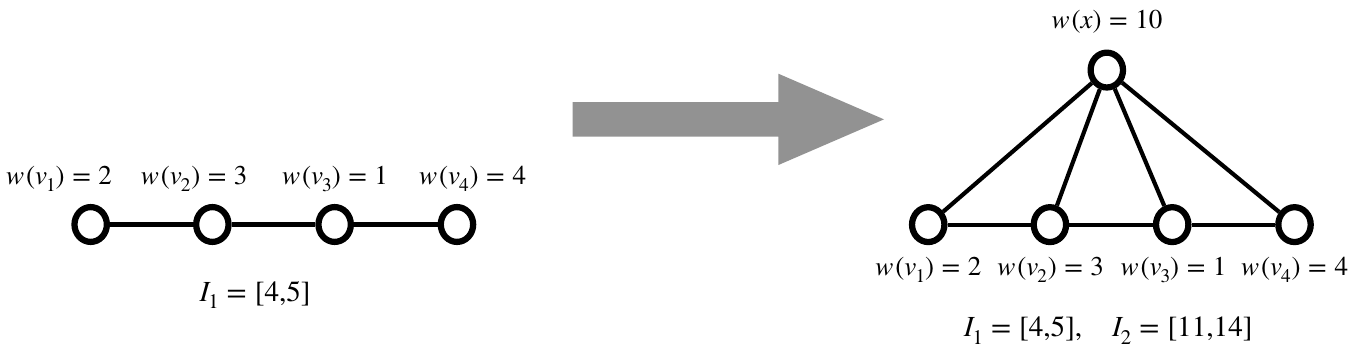}
    \caption{\label{fig:universal_a}}
    \end{subfigure}
    \qquad
    \begin{subfigure}[b]{0.37\textwidth}
    \centering
    \includegraphics[width=\linewidth]{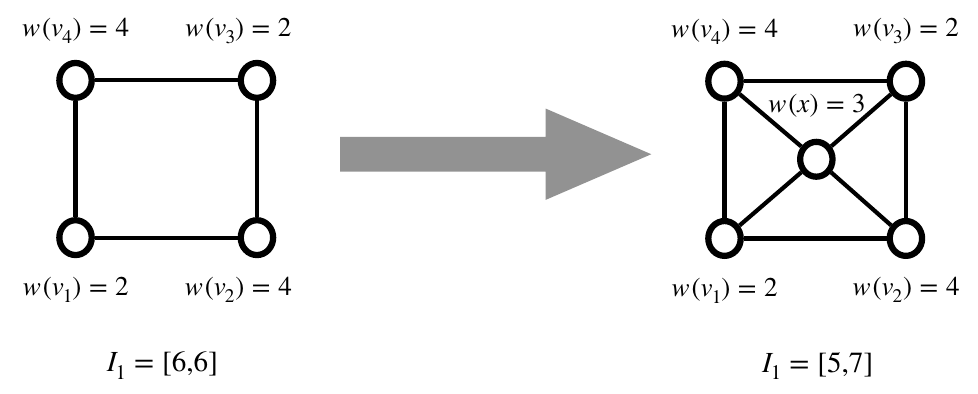}
    \caption{\label{fig:universal_b}}
    \end{subfigure}    
\caption{(\subref{fig:universal_a}) An example of a star-$1$-PCG for which no $1$-witness is free and where the addition of a universal vertex necessarily increases the number of intervals (see the graph $G_{27}$ in \cite{monti2024starkpcgs}) and (\subref{fig:universal_b}) an example of a star-$1$-PCG for which no  $1$-witness is free and where the addition of a universal vertex does not increase the number of intervals.}\label{fig:universal}
\end{figure*}

\begin{theorem}
Given a graph $G$ with $\gamma(G) = k$, let $G'$ be the graph obtained from $G$ by adding a universal vertex. If there exists a $G^w$ that is free then $\gamma(G') = k$, otherwise $\gamma(G') \leq  k+1$. 
\end{theorem}
\begin{proof}
Let $G=(V,E)$ be a graph with star number $k$ and let $G^w$ be a $k$-witness graph with $I_i=[a_i, b_i]$ for $1\leq i \leq k$ and let $G'=(V\cup \{x\},E')$ where $E'=E \cup \{xv: v\in V\}$. First assume there exists $G^w$ that is free. By Lemma~\ref{lem:left-right-unbalanced} we can assume \textit{w.l.o.g.} that $G^w$ is right-free and thus for any non-edge $uv$ it holds $w(uv) < a_k\leq b_k$. This clearly means that for all $u\in V$, $w(u)< b_k$. Then we define the graph $G'^{w'}$ from $G^w$ by extending the weight $w'$ from $w$  as follows
$$
   w'(u) = 
    \begin{cases} 
    b_k +1  & \text{if } u =x, \\
      w(u) & \text{otherwise } 
    \end{cases}
$$
For $1\leq i \leq  k-1$ we set $I'_i=I_i$  and $I'_k=[a_k, 2b_k+1]$.  

For any $uv \in E$ we have that there exists an $i$ such that $w(uv) \in I_i$ and we have $w'(uv)=w(uv) \in I_i \subseteq I'_i$. Next, for any edge $xu$ where $u \in V$, we have $a_k < b_k+1 \leq w'(u)+w'(x)\leq b_k+b_k+1=2b_k+1$ and thus $w'(xu) \in I'_k$.  Finally, for any non-edge $uv \not\in E$, we have that  and thus $w'(uv)=w(uv) \not\in \cup_{1\leq i \leq k-1} I_i = \cup_{1\leq i \leq k-1} I'_i$. Moreover as $G^w$ is right-free  we have $w'(uv)=w(uv) < a_k \not \in I'_k$.  Thus, $G'^{w'}$ is a $k$-witness for $G'$.

Now, assume all the $k$-witnesses of $G$ are not free. In that case let 
$M=\max_{e \not \in E} w(e)$ and we define the weight $w'$ from $w$  as follows
$$
   w'(u) = 
    \begin{cases} 
    M+1  & \text{if } u =x, \\
      w(u) & \text{otherwise } 
    \end{cases}
$$ 

For all $1\leq i \leq  k$ we set $I'_i=I_i$  and $I'_{k+1}=[M, 2M]$. Clearly for any $uv \in E$ as the weights of the vertices and the $k$ intervals  $I_1, \ldots, I_k$ remain the same, it continues to hold that there exists an $1\leq i\leq k$ such that $w(uv) \in I_i$.

Next, for any edge $xu$ where $u \in V$, we have $M \leq w'(xu)=w(u)+M+1\leq 2M +1$ where the  last inequality follows by the fact that $w(u) \leq M$.  Hence, $w'(xu) \in I'_1$. Finally, for any non-edge $uv \not\in E$, we have that first $w'(uv)=w(uv) \not \in \cup_{i=1}^k I_k$  and  $w(uv) < M+1$ and thus $w'(uv) \not \in \cup_{1\leq i\leq k+1} I'_i$. 
This concludes the proof.
\end{proof}

In Figure~\ref{fig:universal_a} we show a path $P_4$, which has no free $1$-witnesses (see Theorem~\ref{theo_path_balanced}), and where the addition of a universal vertex increases the number of intervals. However, in Figure~\ref{fig:universal_b} we have another star-$1$-PCG which  has no free $1$-witnesses (the proof is almost identical to the one in Theorem~\ref{theo_path_balanced}) where  the addition of a universal vertex does not increase the number of intervals.

\begin{theorem}\label{theo:pendant}
Given a graph $G=(V,E)$ with $\gamma(G) = k$, let $X$ be a set of vertices for which $V\cap X =\emptyset$ and let $G'$ be the graph obtained from $G$ by adding the vertices of $X$ as pendant. If there exists a $G^w$ that is free then $\gamma(G') = k$, otherwise $\gamma(G') \leq  k+1$. 
\end{theorem}
\begin{proof}
Let $G=(V,E)$ be a graph with star number $k$ and let $G^w$ be a $k$-witness graph  in normal form with $I_i=[a_i, b_i]$ for $1\leq i \leq k$. Assume first $G^w$ is not free. From Lemma~\ref{lem:distinct_weight} we can assume all the vertices have distinct weights. Denote by $W=\max_{v\in V} w(v)$ and set $M=\max \{W, b_k\} +1$.  Now,  let $G'=(V\cup X ,E')$ obtained from $G$ by adding the pedant vertices in $X$. Then we define the graph $G'^{w'}$ from $G^w$ as follows: 
$$
   w'(u) = 
    \begin{cases} 
     w(u)  & \text{if } u \in V, \\
      2M-w(v) & \text{if } u \in X \text{and } p(u)=v
    \end{cases}
$$
Notice that the weights of the edges in $E$ do not change and moreover each pendant edge has weight exactly $2M$.  For $1\leq i \leq  k$ we set  $I'_i=I_i$ and $I'_{k+1}=[2M, 2M]$.  We show now that $G'^{w'}$ is a $(k+1)$-witness for $G'$.

For any edge $uv$ with $u \in V, v\in V$, we have that there exists an $1\leq i\leq k$ such that $w(uv) \in I_i$ and clearly by definition of $w'$ we have $w'(uv) = w(uv) \in I_i=I'_i$. Next, for any pendant edge $xv$ in $G'$ with $x \in X$ and $p(x)=v \in V$ we have $w'(xv)=2M-w(v)+w(v)=2M\in I'_{k+1}$.  

For any non-edge $uv$, with $u \in V$, $v\in V$ we have $w'(uv) = w(uv) \not \in \cup_{1\leq i\leq k}I'_i$. Moreover, as $w(u)+w(v) < M + M$ (as the weights are all distinct in $G^w$) we have that $w'(uv) \not \in I'_{k+1}$.   

Next, for any non edge $xu$ with $x \in X$ and $u \in V$, we have $w'(xu) = 2M - w(p(x)) + w(u)$. As $p(x) \neq u$ and all the vertices in $G^w$ have distinct weights we have that  $w'(xu) \neq 2M$ and thus it is not in $I'_{k+1}$. Moreover, $w'(xu) = 2M - w(p(x)) + w(u) > M + w(u) > b_k$, where in the first and last inequality we use $M> w(p(x))$ and $M> b_k$, respectively. Hence, $w'(xu) \not \in \cup_{1\leq i\leq k}I'_i $.

Finally, for any non edge $x_1 x_2$ with $x_1 \in X$ and $x_2 \in X$, we have $w'(x_1 x_2)= 4M -w(p(x_1)) - w(p(x_2)) > 4M - M - M =2M$, where the inequality holds as $M$ is strictly greater than any weight of the vertices in $V$. Hence, $w'(x_1 x_2) \not \in \cup_{1\leq i\leq k+1}I'_i $.

\noindent
Thus $G'^{w'}$ is a $(k+1)$-witness for $G'$.

Finally, notice that if there exists a $G^w$ which is free, thanks to Lemma~\ref{lem:left-right-unbalanced} and Corollary~\ref{cor:normal_form_free} we can always assume $G^w$ is normal and right-free.  The proof follows identically as in the previous case noticing that $I'_k$ and $I'_{k+1}$ can be merged together meaning that $G'^{w'}$ is a $k$-witness for $G'$.
\end{proof}

In Figure~\ref{fig:pendant}a we show a path $P_5$, which has no free $1$-witnesses (see Corollary~\ref{corollary_path_balanced}), and  for which the addition of a pendant vertex  does not increase the number of intervals (see Figure~\ref{fig:pendant}b).  Next, from Theorem~\ref{corollary_path_balanced} the caterpillar in Figure~\ref{fig:pendant}b has no free $1$-witness.   By adding another pendant vertex as in Figure~\ref{fig:pendant}c we obtain an example of an \emph{asteroidal triple} $AT$ \cite{Golumbic2004}, for which in  \cite{Jamison2021} it was shown that $\gamma(AT) > 1$ and in \cite{monti2024starkpcgs} a construction was provided concluding that $\gamma(AT) =2$.  Thus, in this case  the addition of a pendant vertex necessarily increases the number of intervals.


\begin{figure*}
\centering
\includegraphics[width=\linewidth]{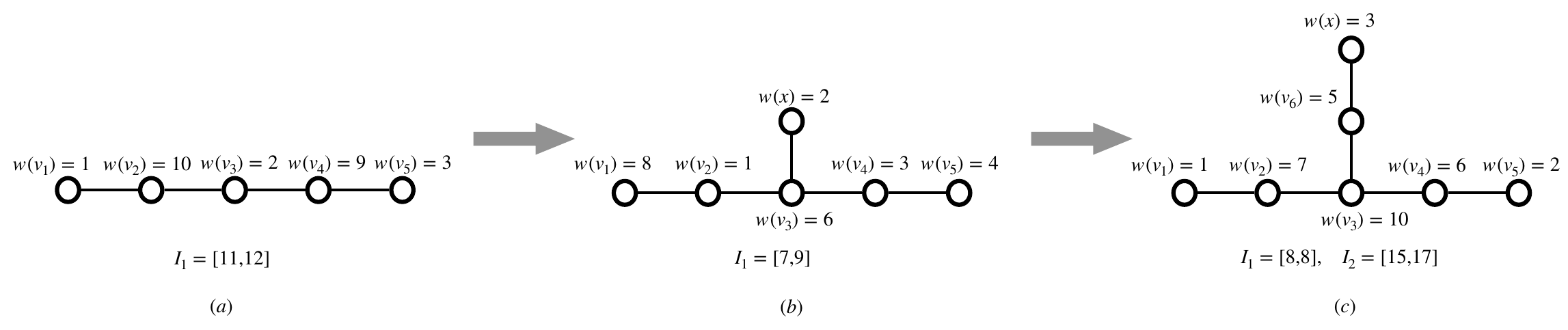} 
\caption{(a)-(b) An example of a star-$1$-PCG for which  no  $1$-witness is free and where the addition of a pedant vertex does not increase the number of intervals (b)-(c) an example of a star-$1$-PCG for which no  $1$-witness is free and where the addition of a pendant vertex increases the number of intervals (see the graph $G_{70}$ in \cite{monti2024starkpcgs}).}\label{fig:pendant}
\end{figure*}

\begin{theorem}\label{theo:twin}
Given a graph $G=(V,E)$ with $\gamma(G) = k$, let $v\in V$ and let $X$ be a set of vertices for which $V\cap X =\emptyset$. For any graph $G'$ obtained from $G$ by adding the vertices of $X$ in such a way that $\{v\} \cup X$ is a set of pairwise false (true) twins in $G'$, it holds $\gamma(G')\leq k+1$.    \end{theorem}
\begin{proof}
 Let $G=(V,E)$ be a graph with star number $k$ and let $G^w$ be a $k$-witness graph in normal form with $I_i=[a_i, b_i]$ for $1\leq i \leq k$.   Fix any $v\in V$ and let  $X$ be such $V\cap X =\emptyset$. We consider first the case of false twins. Hence, let $G'=(V',E')$ such that $V'=V\cup X$ and $E'=E\cup\{xu: x\in X \wedge vu \in E\}$. Then we define $w'$ as follows:
 $$
   w'(u) = 
    \begin{cases} 
     w(u)  & \text{if } u \in V, \\
      w(v) & \text{if } u \in X 
    \end{cases}
$$
Notice that if $2w(v) \not \in \cup_i I_i$ then we could set for all $1\leq i \leq k$, $I'_i=I_i$ and  $G^{w'}$ is clearly a witness for $G'$. Suppose now that there exists  $i$ such  that $2w(v) \in I_i$. Notice that as $G^w$ is normal form there is no edge $e \in E$ such that $w(e)=2w(v)$.  We split the interval $I_i$ into at most two smaller intervals, ensuring that the value $2w(v)$ falls outside these new intervals.  To this purpose, we define 
$$
\epsilon=\min_{u_1, u_2 \in V; \\ w(u_1u_2) \in I_i}  \Big\{  \left| 2w(v) - w(u_1u_2))\right| \Big\}
$$
Note that $\epsilon$ represents the distance of $2w(v)$ from the closest edge of $G$ whose weight falls within $I_i$. Moreover, $\epsilon$ is well-defined as there is at least one edge with weight in $I_i$ (otherwise $\gamma(G)<k$) and $\epsilon>0$ as $G^w$ is in normal form. We consider three cases:
\begin{itemize}
    \item $2w(v) - \epsilon/2 \leq  a_i$. In this case we set the interval $I'_i=[2w(v) + \epsilon/2, b_i]$.
     \item $2w(v) + \epsilon/2 \geq  b_i$. In this case we set the interval $I'_i=[a_i, 2w(v) - \epsilon/2]$.
     \item otherwise, we define two new intervals $I'_i=[a_i, 2w(v) - \epsilon/2]$ and $I''_i=[2w(v) + \epsilon/2, b_i]$.
\end{itemize}
In the first two cases we set the new intervals as $I_1, \ldots, I_{i-1}, I'_{i}, \ldots I_k$ and in the third case as $I_1, \ldots, I_{i-1}, I'_{i}, I''_{i}, \ldots I_k$. In both cases we have at most $k+1$ intervals and it is not difficult to check that $G'^{w'}$ is a $(k+1)$-witness for $G'$.

In the case of true twins we define the $w'$ in the same way as in the case of false twins. Notice that if $2w(v) \in \cup_i I_i$ then we could set for all $1\leq i \leq k$, $I'_i=I_i$ and  $G^{w'}$ is clearly a $k$-witness for $G'$. Otherwise we add the new interval $I=[2w(v), 2w(v)]$.  Notice that as $G^w$ is normal form there is no pair of vertices $u_1, u_2$ such that $w(u_1 u_2)=2w(v)$. It is easy to see that we obtain a $(k+1)$-witness for $G'$.
\end{proof}


\begin{lemma}\label{lem:complement}
Given a graph $G=(V,E)$ with $\gamma(G) = k$, let $\overline{G}$ be its complement. It holds   $|\gamma(\overline{G}) - \gamma(G)| \leq 1$. 
\end{lemma}
\begin{proof}
Assume on the contrary that there exists $G$ and $\overline{G}$ such that $|\gamma(\overline{G}) - \gamma(G)| \geq  2$.  Let $\gamma(G)=k$,  and assume \textit{w.l.o.g.} that $ \gamma(\overline{G}) \geq \gamma(G) +2$. $G^w$ be a $k$-witness graph  with $I_i=[a_i, b_i]$ for $1\leq i \leq k$. We construct $\overline{G}^{w'}$ as follows: $w'=w$ and  $I'_1, \ldots, I'_{k'}$ are the intervals defined by the set  $R^+\setminus \bigcup_i I_i$.  By construction $k'=k+1$ and it is not difficult to see that  ${\overline{G}}^{w'}$ is a $k'$-witness for  $\overline{G}$. Thus, we have that $\gamma(\overline{G})  \leq k'=k+1=\gamma(G)+1$ contradicting our initial hypothesis.
\end{proof}

\begin{theorem} \label{theo:complement}
Given a graph $G=(V,E)$ with $\gamma(G) = k$, let $\overline{G}$ be its complement. It holds
\begin{itemize}
    \item[(i)]  if there exists a $k$-witness $G^w$ that is both left- and right-free then $\gamma(\overline{G}) = k-1$.
    \item[(ii)] otherwise if there exists a $k$-witness $G^w$ that is free  then $\gamma(\overline{G}) \leq k$.
    \item[(iii)] otherwise if no $k$-witness $G^w$  is free  then $\gamma(\overline{G}) \leq k+1$.
\end{itemize}
\end{theorem}

\begin{proof}
The three items can be proved using the same construction described in Lemma~\ref{lem:complement} and observing that: 
 (i) if $G^w$ is both left and right free, then the intervals $I'_1$ or $I'_{k'}$  are not used, that is, there is no edge $e$ of $\overline{G}$ whose weight falls in $I'_1$ or $I'_{k'}$ and thus, $\gamma(\overline{G}) \leq k-1$ and by the result of Lemma~\ref{lem:complement} we have that the equality must hold; (ii) if $G^w$ is free then exactly one among $I'_1$ or $I'_{k'}$ is not used and hence  $\gamma(\overline{G}) \leq k$; (iii) if $G^w$ is not free, then the construction we provide uses all the intervals and hence we can only claim that $\gamma(\overline{G}) \leq k+1$.
\end{proof}

Notice that there are examples of graphs showing that the inequalities in (ii) and (iii) cannot be tight. Every graph that is self-complementary  (\textit{i.e.} a graph which is isomorphic to its  complement) satisfies $\gamma(G)=\gamma(\overline{G})$. Thus, self-complementary graphs that have a free witness are examples for which the inequality in (ii) is tight. An example is shown in Figure~\ref{fig:auto_complementary_free}.
Next, the complement of the cycle on 7 vertices, $\overline{C_7}$, has star number 3 (see \cite{monti2024starkpcgs}) and a right-free $3$-witness is depicted in Figure~\ref{fig:auto_complementary_path}. As the star number of cycles is $2$ \cite{Monti2023} this graph provides an example for which the inequality in (ii) is strict.

Consider now the inequality in (iii).  Self-complementary graphs for which no free witness exists (like $P_4$ and $C_5$) are examples for which the inequality in (iii) is strict.  Moreover, by Theorem~\ref{theo_path_balanced}, $C_7$ has no free $1$-witnesses and thus provides an example for which the inequality in (iii) is tight.

\section{The star number of acyclic graphs}\label{sec:acyclic}
In this section, we show how the results from Section~\ref{sec:operations} can be applied to determine the star number of lobster graphs and more generally, improve the existing bounds on the star number of acyclic graphs.  The following result holds.

\begin{theorem}\cite{Jamison20,monti2024starkpcgs} \label{theo:caterpilalr}
For any caterpillar $T$, it holds $\gamma(T)=1$.
 \end{theorem}

Since $AT$ (see the graph in Figure~\ref{fig:pendant}c) has star number 2 and is a subgraph of any binary tree of radius at least two, we have the following corollary. 

\begin{corollary}\label{cor:binary_tree_lowerbound}
For any binary tree $T$ of radius at least 2 it holds $\gamma(T)\geq 2$.
\end{corollary}

A lobster can be obtained from a caterpillar by adding pendant vertices. Thus, we have the following result.

\begin{theorem}
For any lobster $L$, with radius at least 2 it holds $\gamma(L)=2$. 
\end{theorem}
\begin{proof}
The proof follows by a direct application of Theorem~\ref{theo:caterpilalr},  Theorem~\ref{theo:pendant} and Corollary~\ref{cor:binary_tree_lowerbound}.
\end{proof}

The following theorem shows that for any acyclic graph $G$ it holds $\gamma(G) \leq rad(G)$. This represents an advancement beyond the broadly acknowledged general result stating $\gamma(G) \leq |E|$ (see \cite{ahmed17}).

\begin{theorem}
For any acyclic graph $G$ it holds $\gamma(G)\leq rad(G)$.
\end{theorem}
\begin{proof}
Let $G$ be an acyclic graph. 
To prove the upper bound we provide a $rad(G)$-witness graph.  Let $\mathcal{CC}$ be the  set of connected components of $G$. Let $|\mathcal{CC}|=t$,  recall that $rad(G)=\max_{{CC}_i \in \mathcal{CC}} {rad({CC}_i)}$ and let $a_i$ denote a center of graph ${CC}_i$. 

We start with the graph $G_0=(V_0,E_0)$ where $V_0=\{a_1, \ldots, a_t\}$ and $E_0=\emptyset$. At step 1 we construct $G_1$ by adding all the vertices adjacent to the vertices in $V_0$ as pendant vertices. Thus, from Theorem~\ref{theo:pendant}  $G_1$ is a $1$-witness. At a generic step $i$ of this procedure we construct the graph $G_i$ from the graph $G_{i-1}$ by adding all the vertices adjacent to the vertices in $V_{i-1}$. Clearly these vertices are all pendant (otherwise they would have been added in a previous step) and we could apply theorem Theorem~\ref{theo:pendant} increasing the number of intervals by at most 1. The number of steps we have to do until we reach $G$ is exactly $rad(G)$. Thus, we have  $\gamma(G)\leq rad(G)$.
\end{proof}

\begin{figure}
\centering
    \includegraphics[scale=0.27]{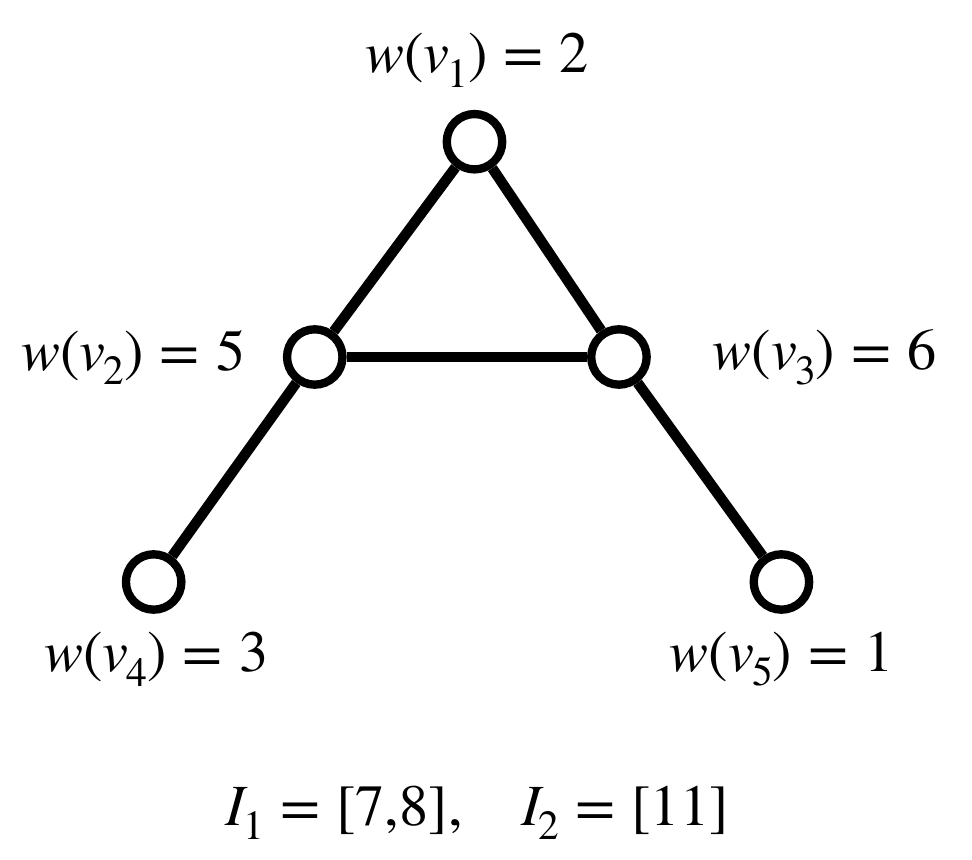}
\caption{A right-free $2$-witness graph of a self-complementary graph $G$ with star number $2$ \cite{monti2024starkpcgs}.}\label{fig:auto_complementary_free}
\end{figure}

\begin{figure}
\centering
    \includegraphics[scale=0.3]{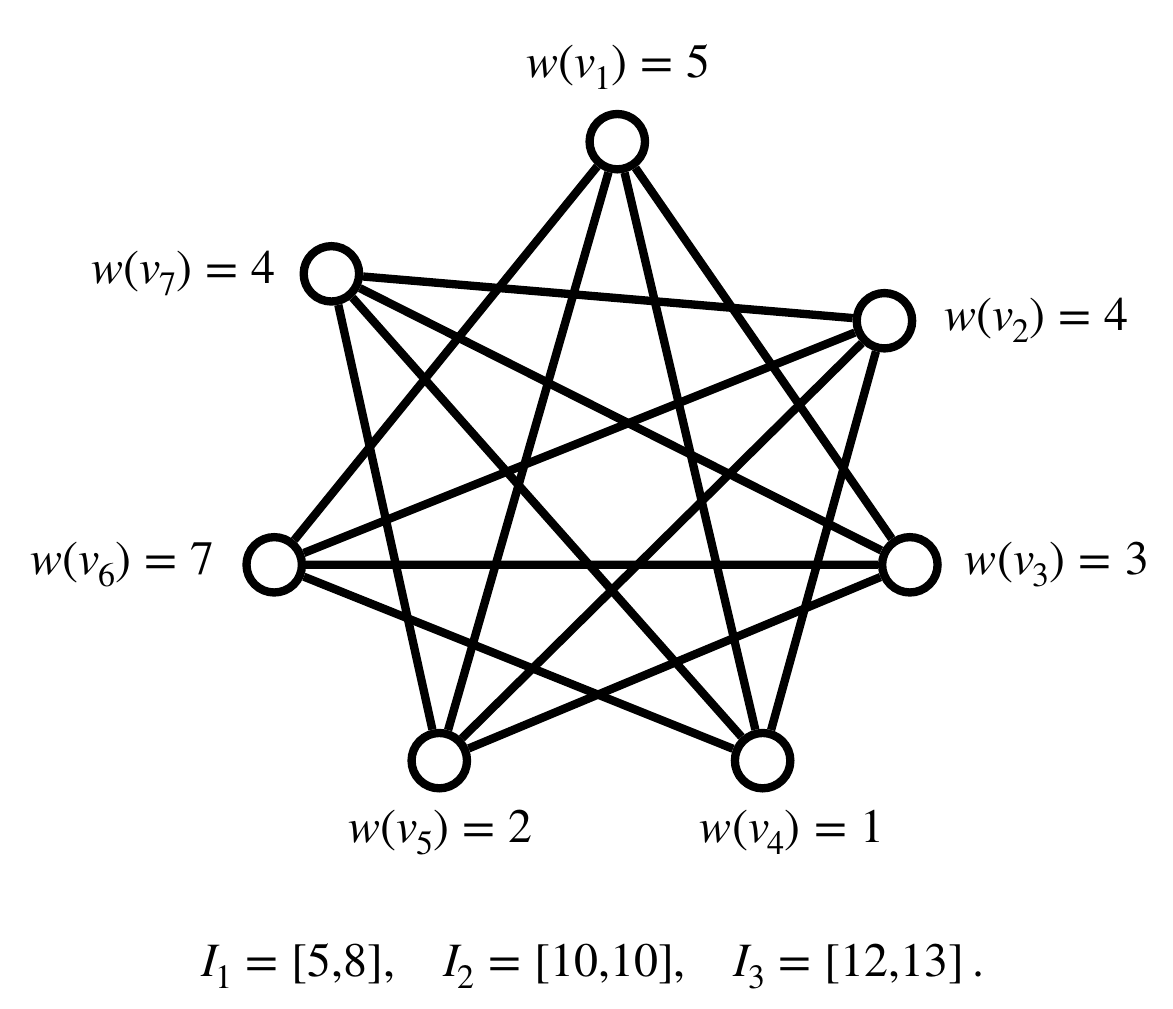}
\caption{A right-free $3$-witness of the graph $\overline{C_7}$ with star number $3$ \cite{monti2024starkpcgs}.}\label{fig:auto_complementary_path}
\end{figure}

\section{Conclusions and open problems}
In this work we study how different graph operations, including the insertion of twins, pendant vertices, universal vertices, and isolated vertices, affect the graph's star number. Then we use these results to compute the star number for lobsters and establish an upper bound for the star number of acyclic graphs. Many problems remain open. 

Firstly, it is interesting to improve the lower bound for the star number of acyclic graphs. Indeed, the only lower bound we have is the trivial value of $2$, which applies to all graphs containing an asteroidal triple.   Next it is interesting to determine whether there exists some constant $c$ such that all acyclic graphs are star-$c$-PCGs.

Finally, it may be worth to consider the effects of other classical graph operations such as the disjoint union of two graphs. While it is not difficult to see that $\gamma(G_1+G_2) \leq \gamma(G_1) + \gamma(G_2)$ improving this result is certainly interesting. Another operation of interest is the cartesian product of two graphs $\gamma(G_1 \times G_2)$. 
   



\section*{Availability of Data and Materials}
No new data were generated or analysed in support of this research.

\section*{Funding}
This research was supported in part by the project "EXPAND: scalable algorithms for EXPloratory Analyses of heterogeneous and dynamic Networked Data", funded by MUR Progetti di Ricerca di Rilevante Interesse Nazionale (PRIN) Bando 2022 - Grant ID 2022TS4Y3N.

\nocite{*}

\bibliographystyle{plain}
\bibliography{my_references}

\begin{thebibliography}{10}

\bibitem{acharya1981}
BD~Acharya and MK~Gill.
\newblock On the index of gracefulness of a graph and the gracefulness of
  two-dimensional square lattice graphs.
\newblock {\em Indian Journal of Mathematics}, 23(81-94):14, 1981.

\bibitem{ahmed17}
Shareef Ahmed and Md.~Saidur Rahman.
\newblock Multi-interval pairwise compatibility graphs.
\newblock In {\em Theory and Applications of Models of Computation}, page
  71–84. Springer International Publishing, 2017.

\bibitem{Calamoneri2012}
T.~Calamoneri, D.~Frascaria, and B.~Sinaimeri.
\newblock All graphs with at most seven vertices are pairwise compatibility
  graphs.
\newblock {\em The Computer Journal}, 56(7):882--886, July 2012.

\bibitem{Calamoneri2013}
T.~Calamoneri, R.~Petreschi, and B.~Sinaimeri.
\newblock On the pairwise compatibility property of some superclasses of
  threshold graphs.
\newblock {\em Discrete Mathematics, Algorithms and Applications},
  05(02):1360002, June 2013.

\bibitem{calamoneri2022}
Tiziana Calamoneri, Manuel Lafond, Angelo Monti, and Blerina Sinaimeri.
\newblock On generalizations of pairwise compatibility graphs, 2021.
\newblock \url{https://arxiv.org/abs/2112.08503}.

\bibitem{Calamoneri2023}
Tiziana Calamoneri, Angelo Monti, and Fabrizio Petroni.
\newblock All graphs with at most 8 nodes are 2-interval-pcgs, 2022.

\bibitem{Calamoneri2016}
Tiziana Calamoneri and Blerina Sinaimeri.
\newblock Pairwise compatibility graphs: A survey.
\newblock {\em {SIAM} Review}, 58(3):445--460, January 2016.

\bibitem{chen2022}
Guantao Chen and Yanli Hao.
\newblock Multithreshold multipartite graphs.
\newblock {\em Journal of Graph Theory}, 100(4):727–732, February 2022.

\bibitem{isgci}
H.~N. de~Ridder et~al.
\newblock
  {{I}}nformation~{S}ystem~on~{G}raph~{C}lasses~and~their~{I}nclusions~({I}{S}{G}{C}{I}),
  n.b.
\newblock \url{https://www.graphclasses.org/smallgraphs.html\#nodes5}.

\bibitem{Durocher2015}
Stephane Durocher, Debajyoti Mondal, and Md.~Saidur Rahman.
\newblock On graphs that are not {PCGs}.
\newblock {\em Theoretical Computer Science}, 571:78--87, March 2015.

\bibitem{Erds}
P.~Erd\"os and G.~Szekeres.
\newblock A combinatorial problem in geometry.
\newblock {\em Compositio Mathematica}, 2:463--470, 1935.

\bibitem{Golumbic2004}
Martin~Charles Golumbic and Ann~N. Trenk.
\newblock {\em Tolerance Graphs}.
\newblock Cambridge University Press, February 2004.

\bibitem{hakim2022}
Sheikh~Azizul Hakim, Bishal~Basak Papan, and Md.~Saidur Rahman.
\newblock New results on pairwise compatibility graphs.
\newblock {\em Information Processing Letters}, 178:106284, November 2022.

\bibitem{Jamison20}
Robert~E. Jamison and Alan~P. Sprague.
\newblock Multithreshold graphs.
\newblock {\em Journal of Graph Theory}, 94(4):518–530, February 2020.

\bibitem{Jamison2021}
Robert~E. Jamison and Alan~P. Sprague.
\newblock Double-threshold permutation graphs.
\newblock {\em Journal of Algebraic Combinatorics}, 56(1):23–41, June 2021.

\bibitem{KPM03}
Paul Kearney, J.~Ian Munro, and Derek Phillips.
\newblock Efficient generation of uniform samples from phylogenetic trees.
\newblock In Gary Benson and Roderic D.~M. Page, editors, {\em Algorithms in
  Bioinformatics}, pages 177--189, Berlin, Heidelberg, 2003. Springer Berlin
  Heidelberg.

\bibitem{Kittipassorn2024}
Teeradej Kittipassorn and Thanaporn Sumalroj.
\newblock Multithreshold multipartite graphs with small parts.
\newblock {\em Discrete Mathematics}, 347(7):113979, July 2024.

\bibitem{Kobayashi22}
Yusuke Kobayashi, Yoshio Okamoto, Yota Otachi, and Yushi Uno.
\newblock Linear-time recognition of double-threshold graphs.
\newblock {\em Algorithmica}, 84(4):1163–1181, 2022.

\bibitem{Long2020}
Yangjing Long and Peter~F. Stadler.
\newblock Exact-2-relation graphs.
\newblock {\em Discrete Applied Mathematics}, 285:212--226, 2020.

\bibitem{monti_ictcs23}
Angelo Monti and Blerina Sinaimeri.
\newblock On graphs that are not star-k-pcgs (short paper).
\newblock In Giuseppa Castiglione and Marinella Sciortino, editors, {\em
  Proceedings of the 24th Italian Conference on Theoretical Computer Science,
  Palermo, Italy, September 13-15, 2023}, volume 3587 of {\em {CEUR} Workshop
  Proceedings}, pages 92--97. CEUR-WS.org, 2023.

\bibitem{Monti2023}
Angelo Monti and Blerina Sinaimeri.
\newblock On star-multi-interval pairwise compatibility graphs.
\newblock In {\em {WALCOM}: Algorithms and Computation}, pages 267--278.
  Springer Nature Switzerland, 2023.

\bibitem{monti2024starkpcgs}
Angelo Monti and Blerina Sinaimeri.
\newblock On star-$k$-pcgs: Exploring class boundaries for small $k$ values,
  2024.
\newblock \url{https://arxiv.org/abs/2209.11860}.

\bibitem{oeis_nr}
{OEIS Foundation Inc.}
\newblock Number of graphs on n unlabeled nodes. entry {A000088}, the
  {O}n-{L}ine {E}ncyclopedia of {I}nteger {S}equences, n.b.
\newblock \url{https://oeis.org/A000088}.

\bibitem{Papan2022}
Bishal~Basak Papan, Protik~Bose Pranto, and Md~Saidur Rahman.
\newblock On 2-interval pairwise compatibility properties of two classes of
  grid graphs.
\newblock {\em The Computer Journal}, 66(5):1256–1267, February 2022.

\bibitem{Puleo2020}
Gregory~J. Puleo.
\newblock Some results on multithreshold graphs.
\newblock {\em Graphs and Combinatorics}, 36(3):913--919, apr 2020.

\bibitem{Rahman2020}
Md.~Saidur Rahman and Shareef Ahmed.
\newblock A survey on pairwise compatibility graphs.
\newblock {\em {AKCE} International Journal of Graphs and Combinatorics},
  17(3):788--795, June 2020.

\bibitem{Sammi2013}
{Sammi Abida} {Salma}, {Md. Saidur} {Rahman}, and {Md. Iqbal} {Hossain}.
\newblock Triangle-free outerplanar 3-graphs are pairwise compatibility graphs.
\newblock {\em Journal of Graph Algorithms and Applications}, 17(2):81--102,
  2013.

\bibitem{Xiao2020}
Mingyu Xiao and Hiroshi Nagamochi.
\newblock Characterizing {Star-PCGs}.
\newblock {\em Algorithmica}, 82(10):3066–3090, oct 2020.

\bibitem{Xiao2020op}
Mingyu Xiao and Hiroshi Nagamochi.
\newblock Some reduction operations to pairwise compatibility graphs.
\newblock {\em Information Processing Letters}, 153:105875, January 2020.

\bibitem{Yanhaona09}
Muhammad~Nur Yanhaona, K.~S. M.~Tozammel Hossain, and M.~Saidur Rahman.
\newblock Pairwise compatibility graphs.
\newblock {\em Journal of Applied Mathematics and Computing},
  30(1–2):479–503, November 2008.

\end{thebibliography}

\end{document}